\documentclass[letterpaper,10pt,conference]{ieeeconf}
\IEEEoverridecommandlockouts
\PassOptionsToPackage{draft}{hyperref}

\usepackage{packages}
\usepackage{theoremstyles}
\usepackage{mitmalstyle}
\usepackage{flushend}

\newcommand{\R}{\mathbb{R}}

\newcommand{\trans}{^{\top}}
\newcommand{\grad}{\nabla}
\newcommand{\half}{\tfrac{1}{2}}
\newcommand{\Msym}{\operatorname{sym}}

\newcommand{\eig}{\operatorname{eig}}
\newcommand{\Real}{\operatorname{Re}}
\newcommand{\Imag}{\operatorname{Im}}

\newcommand{\blockdiag}{\operatorname{blockdiag}}
\newcommand{\tr}{\operatorname{tr}}

\newcommand{\Metr}{M}
\newcommand{\Aanti}{A}
\newcommand{\Ssym}{S}
\newcommand{\ths}{\theta^\star_{s}}
\newcommand{\thr}{\theta^\star_{r}}
\newcommand{\Ms}{\Metr_{s}}
\newcommand{\Mr}{\Metr_{r}}

\newcommand{\Jhat}{\widehat J}

\newcommand{\Jflip}{\mathcal{J}_{\mathrm{flip}}}

\renewcommand{\baselinestretch}{0.985
}
\begin{document}

\title{Metric-Based Equilibrium Selection in Noncooperative Differentiable Games}

\author{Karan~Mahesh$^{1,2}$, Runyu~Zhang$^{3}$, Michael~R.~Benjamin$^{1}$, and Gioele~Zardini$^{3}$%
    \thanks{This material is based upon work supported by the National Science Foundation Graduate Research Fellowship Program under Grant No.~2141064. Any opinions, findings, and conclusions or recommendations expressed in this material are those of the author(s) and do not necessarily reflect the views of the National Science Foundation.}
    \thanks{$^{1}$K.~Mahesh and M.~R.~Benjamin are with the Department of   Mechanical Engineering, Massachusetts Institute of Technology, Cambridge, MA, USA. {\tt\small \{karanmah, mikerb\}@mit.edu}}%
    \thanks{$^{2}$K.~Mahesh is with the Woods Hole Oceanographic Institution (WHOI), Woods Hole, MA, USA.}%
    \thanks{$^{3}$R.~Zhang and G.~Zardini are with the Laboratory for Information and Decision Systems, Massachusetts Institute of Technology Cambridge, MA USA.{\tt\small \{runyuzha, gzardini\}@mit.edu}}%
}

\maketitle

\begin{abstract}
Metric conditioning can change which equilibrium attracts learning dynamics
without moving the equilibria of a differentiable game. We study how to use
state-dependent symmetric positive-definite (SPD) metrics to select among known
equilibria by keeping a designated equilibrium attracting while making a rival
unstable. Building on classical results on matrix stability under
positive-definite multiplication, we characterize when such destabilization is
possible and show how restricting the metric to act independently on each
player limits the stability changes that can be achieved at differential Nash
equilibria. We then combine a stabilizing metric at the selected
equilibrium with a destabilizing metric at the rival through smooth
interpolation, yielding a single metric field that preserves every equilibrium of the game while assigning the desired local stability types. We also derive a step-size condition for implementing the method in discrete
time. The approach is validated on a
continuous-commitment stag hunt and an entropy-regularized iterated prisoner's
dilemma, where it selects the desired equilibrium in both settings.
\end{abstract}

\section{Introduction}
\label{sec:intro}

Interaction-aware game solvers, including iterative linear-quadratic games and
their constrained variants~\cite{fridovichkeil2020ilqgames}, converge to a
local equilibrium determined by the basin in which they start. When several
collision-free equilibria exist, as at unsignalized intersections or in
head-on maritime encounters, the computed outcome can therefore depend on
initialization. Existing systems may resolve this ambiguity after solving the
game by applying right-of-way rules~\cite{avselection2024}. Our objective is to
select a specific equilibrium without modifying the game's equilibria themselves. We accomplish this with a change in the underlying geometry of the players' learning updates. Rescaling and reorienting the update directions leaves every
equilibrium in place, because at an equilibrium every player's own gradient is
already zero, but it can change which equilibria the learning process is drawn
to.
This raises a structural question: at which equilibria can such a change of
geometry make the equilibrium either attracting or unstable? We show that this
depends on the local game structure, not merely on whether the players
cooperate, and on whether the geometry may couple the players' updates or must
act on each player separately.

In this paper, we use a state-dependent symmetric positive-definite (SPD)
metric to precondition the game dynamics without moving its equilibria. Given two attracting equilibria, we designate one as the selected
equilibrium and the other as the rival. We choose a metric $M_s$
that keeps the selected equilibrium attracting and a metric $M_r$
that makes the rival unstable, and smoothly interpolate between them to obtain
a single metric field. We characterize when such destabilization is possible,
how block-diagonal metric restrictions limit it, and how the local stability
assignment carries over to discrete-time updates.

\paragraph{Related work}
Equilibrium selection for differentiable games has significant prior work. Young~\cite{young1993conventions}
characterizes the conventions selected by rare mistakes through long-run
stochastic stability. Here, by contrast, the desired convention is specified
in advance and encoded in deterministic learning dynamics. Zhang et
al.~\cite{zhang2025steer} study how to steer an unknown no-regret learner
toward a Stackelberg outcome; their information requirements and unilateral
influence differ from the design of a joint metric at known zeros of the game
field.

Metric-based game dynamics have precedent. Mertikopoulos and
Sandholm~\cite{mertikopoulos2018riemannian} derive Riemannian dynamics for
population games that encompass the replicator and projection dynamics. We employ similar constructions of game dynamics to determine which local Jacobians admit a change of
stability over the cone of symmetric positive-definite (SPD) matrices. Modifying these gradient dynamics to influence their limit is also established.
Shamma and Arslan~\cite{shamma2005dynamic} introduce dynamic fictitious and
gradient play, and Chasnov et al.~\cite{chasnov2020convergence} show that
unequal learning rates can change stability and basins of attraction while
preserving stationary points. Nash optimality and dynamical stability are
distinct properties~\cite{mazumdar2020gradient}. Local symplectic
surgery~\cite{mazumdar2025surgery} makes differential Nash equilibria the
attracting equilibria of an adjusted zero-sum flow under its assumptions. These methods, however, do not explicitly consider the existence of Nash equilibria that would converge under typical gradient methods but are not desirable. In our method, the aim is
to remove attraction to rival equilibria while keeping its location.

The decomposition of the game Jacobian into symmetric and antisymmetric parts
\cite{balduzzi2018mechanics,letcher2019differentiable} exposes the local
structure used here. SGA~\cite{balduzzi2018mechanics} and consensus
optimization~\cite{mescheder2017consensus} modify the field through its
Jacobian, but with fixed gains they treat symmetry-related equilibria
identically. Natural-gradient and K-FAC
methods~\cite{amari1998natural,martens2015kfac} also precondition by a
metric, but their symmetric statistics need not identify the full game
Jacobian.

Results on $H$-stability and inertia characterize when positive-definite
multipliers preserve stability properties
\cite{ostrowski1962inertia,carlson1963inertia,carlson1968hstability,kushel2019stability}, while
~\cite{johnson1974field} relates the field of values of a matrix to
eigenvalues attainable under positive-definite left multiplication. These
results already settle the single-equilibrium question: the necessary
conditions of \Cref{thm:necessary} specialize the inertia and $H$-stability
theory~\cite{ostrowski1962inertia,carlson1963inertia,carlson1968hstability},
and the existence of a destabilizing metric in \Cref{thm:switch} follows
directly from~\cite[Thm.~1$'$]{johnson1974field}. The contribution here is to carry these results over to
differentiable games: how to use such multipliers at different equilibria of the same game, subject to
player-block structure, and combine them into a single state-dependent metric
that preserves the original equilibrium set. Interpolating between the two metrics is
also related to gain-scheduled control~\cite{stilwell2000interpolation},
although here the objective is to assign different local stability types to
specified zeros rather than stabilize an operating family.

The main contributions of this paper are as follows:
\begin{enumerate}
    \item Building on classical results on matrix stability under positive-definite
          multiplication~\cite{ostrowski1962inertia,carlson1963inertia,johnson1974field},
          we characterize when an attracting rival equilibrium can be made unstable by
          a metric $M_r$, and show how restricting the metric to act
          independently on each player limits the stability changes that can be achieved.

    \item We combine the metrics $M_s$ and $M_r$ into a smooth
          state-dependent metric that preserves all equilibria while keeping
          the selected equilibrium attracting and making the rival unstable.
          We also establish when this local stability assignment carries over
          to discrete-time updates.

    \item We demonstrate the approach on a stag hunt and an iterated prisoner's
          dilemma, where it selects the payoff-dominant equilibrium and a
          prescribed player role, respectively, while preserving the original
          game equilibria.
\end{enumerate}
\section{Problem formulation}
\label{sec:prelim}

Each of $n$ players $i\in\{1,\dots,n\}$ holds a policy $\theta_i\in\R^{d_i}$
and minimizes a loss $L_i(\theta)$, where
$\theta=(\theta_1,\dots,\theta_n)\in\R^m$ is the joint policy and
$m=\sum_i d_i$.
Stacking each player's gradient with respect to its own policy gives the
\emph{simultaneous gradient}~\cite{balduzzi2018mechanics} and its Jacobian,
\begin{equation}
    \begin{aligned}
        \xi(\theta) & =\operatorname{col}_i(\grad_{\theta_i}L_i) \\
        J(\theta)   & =D\xi(\theta)=[J_{ij}(\theta)]_{i,j},\quad
        J_{ij}      =\nabla_{\theta_j}\nabla_{\theta_i}L_i.
    \end{aligned}
    \label{eq:field}
\end{equation}
Row block $i$ collects the derivatives of player $i$'s own gradient. The
diagonal blocks are the players' own Hessians and are symmetric; an
off-diagonal block $J_{ij}$ measures how player $j$'s decision alters player
$i$'s marginal incentive. Because $J_{ij}$ and $J_{ji}$ derive from different
losses, in general $J_{ij}\ne J_{ji}^\top$. Following Balduzzi et
al.~\cite{balduzzi2018mechanics}, write
\begin{equation}
    J=\Ssym+\Aanti,\qquad
    \Ssym=\half(J+J\trans),\quad \Aanti=\half(J-J\trans).
    \label{eq:decomposition}
\end{equation}
Equilibria are the zeros $\xi(\theta^\star)=0$. In simultaneous gradient play,
$\theta_{k+1}=\theta_k-\eta\xi(\theta_k)$, every player descends its own loss
while treating the others' decisions as fixed.

A \emph{metric field} is a
$C^1$ map $\theta\mapsto\Metr(\theta)\in\R^{m\times m}$ whose values are
symmetric positive definite (SPD), written $\Metr(\theta)\succ0$. At each
joint policy $\theta$ it defines the inner product
$\langle u,v\rangle_{\Metr(\theta)}=u^\top\Metr(\theta)v$ between tangent
vectors $u,v\in\R^m$. Preconditioning the simultaneous gradient by a metric field gives the
preconditioned simultaneous-gradient flow
\begin{equation}
    \dot\theta=-\Metr(\theta)^{-1}\xi(\theta).
    \label{eq:flow}
\end{equation}
Since $\xi(\theta^\star)=0$, every $\Metr\succ0$ leaves every equilibrium fixed.
Unequal positive learning rates are a special case: the update
$\theta_{k+1}=\theta_k-\Gamma\xi(\theta_k)$ with
$\Gamma=\blockdiag(\gamma_i I_{d_i})\succ0$ is the explicit Euler step of
\eqref{eq:flow} with step size $\eta$ and constant metric
$\Metr=\eta\Gamma^{-1}$~\cite{chasnov2020convergence}.
Linearizing with $\delta=\theta-\theta^\star$ gives
$\dot\delta=-\Metr(\theta^\star)^{-1}J(\theta^\star)\,\delta$. Below, $\Metr$ and $J$
denote these values at the equilibrium. The selection criterion uses the
full-spectrum \emph{stability margin}
\begin{equation}
    \mu(\Metr,J):=\min_i\Real\lambda_i(\Metr^{-1}J).
    \label{eq:margin}
\end{equation}
A positive margin implies local exponential stability and a negative margin
implies linear instability; at zero margin the linearization is
inconclusive. An equilibrium is \emph{nondegenerate} if its Jacobian is
invertible, $\det J(\theta^\star)\ne0$; it is then isolated, and because
$\det(\Metr^{-1}J)=\det J/\det\Metr$, nondegeneracy does not depend on the
metric. A linearization is \emph{hyperbolic} if no eigenvalue of
$\Metr^{-1}J$ has zero real part, and its \emph{unstable dimension} is the
number of eigenvalues with negative real part. Under $\Metr$, the equilibrium
is a \emph{strict sink} if $\mu(\Metr,J)>0$, a
\emph{source} if every eigenvalue has negative real part (unstable dimension
$m$), and a \emph{saddle} if the linearization is hyperbolic with unstable
dimension strictly between $0$ and $m$.

For the linearized gradient flow $\dot\delta=-J\delta$,
\begin{equation}
    \frac{d}{dt}\half\|\delta\|^2=-\delta^\top J\delta
    =-\delta^\top\Ssym\delta,
    \label{eq:energy}
\end{equation}
since $\delta^\top\Aanti\delta=0$. Thus $\Ssym$ determines the instantaneous
rate of change of the Euclidean distance to the equilibrium, whereas $\Aanti$
only rotates the state without directly changing that distance.

If $\Aanti(\theta)\equiv0$ on a simply connected open domain, the field is a
gradient, $\xi=\grad\Phi$, the game is a potential game, and \eqref{eq:flow}
is preconditioned gradient descent on $\Phi$. In potential games, metrics may
change basins without changing the linear stability type of nondegenerate
equilibria. Following Ratliff et al.~\cite{ratliff2016characterization}, a zero
$\theta^\star$ is a \emph{differential Nash equilibrium} if each player's own
Hessian is positive definite there:
\begin{equation}
    \xi(\theta^\star)=0,\qquad
    J_{ii}(\theta^\star)=\grad^2_{\theta_i}L_i(\theta^\star)\succ0
    \quad\text{for every }i.
    \label{eq:differentialnash}
\end{equation}
These own Hessians are the diagonal blocks $J_{ii}=\Ssym_{ii}$. Local Nash
optimality thus concerns only the own blocks, whereas stability of gradient
play depends on the full $J$. As Mazumdar et al.~\cite{mazumdar2020gradient}
show, attracting zeros can fail to be Nash, and differential Nash equilibria
can be saddles of gradient play. Below, \emph{equilibrium} means a zero of
the flow unless Nash optimality is stated explicitly.

\paragraph{Selection problem}
Given two equilibria, a \emph{selected} $\ths$ and a \emph{rival} $\thr$, find
$\Metr(\theta)\succ0$ such that
\begin{equation}
    \mu\big(\Metr(\ths),J(\ths)\big)>0
    \quad\text{and}\quad
    \mu\big(\Metr(\thr),J(\thr)\big)<0,
    \label{eq:selection}
\end{equation}
so that $\ths$ is locally attracting and $\thr$ is not, while neither
equilibrium moves. Typically the rival is the image of $\ths$ under a
symmetry of the game, and both are sinks under $\Metr=I$.

\section{When can a metric change equilibrium stability?}
\label{sec:structure}

When the rival is initially a strict sink, selection in the sense of
\eqref{eq:selection} requires a positive-definite metric that makes it
unstable. Let
\begin{equation}
    \Jflip=\big\{J:\ \exists\,\Metr_{+},\Metr_{-}\succ0,\ \mu(\Metr_{+},J)>0>\mu(\Metr_{-},J)\big\}
    \label{eq:jflip}
\end{equation}
denote the class of Jacobians that some SPD metric stabilizes and another destabilizes:
one metric, $\Metr_{+}$, makes the equilibrium a sink and another, $\Metr_{-}$,
makes it non-attracting.
The single-equilibrium facts below are specializations of classical results
on inertia and $H$-stability
\cite{ostrowski1962inertia,carlson1963inertia,carlson1968hstability,johnson1974field};
they are restated in the present notation, with short proofs for
completeness, because the selection construction of \Cref{sec:statedep}
uses them at each equilibrium. The player-block restrictions of
\Cref{thm:nash} then identify the role of the local game structure. The characterization holds equilibrium by equilibrium,
so different zeros of the same game can differ in whether they can be
destabilized.

The \emph{inertia} of a symmetric matrix $X$ is the triple
$\operatorname{In}(X)=(n_+,n_-,n_0)$ of its numbers of positive, negative, and
zero eigenvalues.

\begin{lemma}[Necessary conditions for SPD destabilization~\cite{ostrowski1962inertia,carlson1963inertia,carlson1968hstability}]
    \label{thm:necessary}
    Let $J=\Ssym+\Aanti$ with $\Ssym=\Ssym\trans$ and $\Aanti=-\Aanti\trans$.
    \begin{enumerate}
        \item If $\Aanti=0$, then for every $\Metr\succ0$ the spectrum of
              $\Metr^{-1}J$ is real, and its numbers of positive, negative, and
              zero eigenvalues equal those of $\Ssym$.
        \item If $\Ssym\succeq0$ then $\Real\eig(\Metr^{-1}J)\ge0$ for every
              $\Metr\succ0$; if $\Ssym\preceq0$ then $\Real\eig(\Metr^{-1}J)\le0$.
    \end{enumerate}
    Consequently, $J\in\Jflip$ implies $\Aanti\neq0$ \emph{and} $\Ssym$
    indefinite. These conditions are not sufficient.
\end{lemma}
\begin{proof}
    With $\Metr^{1/2}$ the symmetric square root, $\Metr^{-1}J$ is similar to
    $\tilde J:=\Metr^{-1/2}J\Metr^{-1/2}$, whose symmetric and skew parts are the
    congruences $\Metr^{-1/2}\Ssym\Metr^{-1/2}$ and
    $\Metr^{-1/2}\Aanti\Metr^{-1/2}$. By Sylvester's law of inertia,
    $\Msym(\tilde J)$ has the inertia of $\Ssym$ for every $\Metr\succ0$.
    (1) If $\Aanti=0$, $\tilde J$ is symmetric, so its spectrum is real with
    that inertia. (2) Every eigenvalue lies in the field of values, so
    \begin{equation}
        \lambda_{\min}\big(\Msym(\tilde J)\big)\le\Real\lambda\le
        \lambda_{\max}\big(\Msym(\tilde J)\big)
        \label{eq:fov}
    \end{equation}
    for all $\lambda\in\eig(\Metr^{-1}J)$; if $\Ssym\succeq0$ the left bound is
    nonnegative, and symmetrically for $\Ssym\preceq0$. For the consequence:
    $\mu(\Metr_{-},J)<0$ excludes $\Ssym\succeq0$ by (2), $\mu(\Metr_{+},J)>0$ excludes
    $\Ssym\preceq0$, so $\Ssym$ is indefinite; and by (1) $\Aanti=0$ would fix
    the number of positive eigenvalues, contradicting $\mu(\Metr_{+},J)>0>\mu(\Metr_{-},J)$.
    Non-sufficiency: take
    $J=\begin{psmallmatrix}1&1/2\\-1/2&-1\end{psmallmatrix}$.
    Its symmetric part is indefinite and its skew part is nonzero, but
    $\det(M^{-1}J)=\det J/\det M<0$, precluding positive stability.
\end{proof}

The necessary conditions become sufficient once a stabilizing metric is
known.
\begin{proposition}[Existence of a destabilizing SPD metric~{\cite[Thm.~1$'$]{johnson1974field}}]
    \label{thm:switch}
    Suppose there exists $\Metr_{+}\succ0$ with $\mu(\Metr_{+},J)>0$. Then a metric
    $\Metr_{-}\succ0$ with $\mu(\Metr_{-},J)<0$ exists if and only if $\Ssym$ is indefinite.
    For a unit vector $v$ with $v^\top\Ssym v<0$, one construction is
    \[
        Q_t=I+t vv^\top,\qquad \Metr_{-}=Q_t^{-1},
    \]
    where $t\ge0$ is large enough that $\tr J+t v^\top\Ssym v<0$.
\end{proposition}
\begin{proof}
    Necessity follows from \Cref{thm:necessary}. For sufficiency,
    $Q_t\succ0$ and
    $\tr(Q_tJ)=\tr J+t v^\top\Ssym v<0$, since
    $v^\top\Aanti v=0$. The real parts of the eigenvalues sum to this trace,
    so at least one is negative.
\end{proof}
In particular, an initially strict sink can be strictly destabilized if and
only if $\Ssym$ is indefinite.

Stability under every SPD, diagonal, or block-SPD multiplier corresponds to
the classical notions of $H$-, $D$-, and block $H$-stability
\cite{kushel2019stability}. The result above characterizes destabilization given a
stabilizing metric; it does not characterize SPD stabilizability in every
dimension.

\begin{remark}[Unstable modes under SPD metrics]
    \label{thm:index}
    At a nondegenerate equilibrium,
    $
        \det(M^{-1}J)=\frac{\det J}{\det M}\neq 0
    $
    for every $M\succ 0$. Consider any two SPD metrics and continuously
    interpolate between them. A real eigenvalue of $M^{-1}J$ cannot cross
    the imaginary axis through zero, while any crossing away from zero
    occurs as a complex-conjugate pair. Hence the number of unstable modes
    can only change by an even number. In particular, starting from a strict
    sink, an SPD metric can only produce an equilibrium with an even number
    of unstable modes, provided no eigenvalue lies on the imaginary axis at
    the endpoint.
\end{remark}
\Cref{thm:nash} characterizes destabilizability through indefiniteness of the symmetric part $S$. The results in \Cref{sec:statedep} characterizes structural conditions that enable or preclude such destabilization.
\begin{theorem}[Indefiniteness at differential Nash equilibria]
    \label{thm:nash}
    Partition $J$ into the player blocks $J_{ij}\in\R^{d_i\times d_j}$ of
    \eqref{eq:field}, and $\Ssym$ and $\Aanti$ accordingly, so that
    \[
        \Ssym_{ij}=\half\bigl(J_{ij}+J_{ji}\trans\bigr),\qquad
        \Aanti_{ij}=\half\bigl(J_{ij}-J_{ji}\trans\bigr).
    \]
    The off-diagonal blocks $\Ssym_{ij}$ and $\Aanti_{ij}$, $i\ne j$, are the
    cross-player parts of the potential and Hamiltonian components $\Ssym$
    and $\Aanti$ of Balduzzi et al.~\cite{balduzzi2018mechanics}.
    \begin{enumerate}
        \item Suppose $\Ssym_{ij}=0$ for all $i\ne j$, that is,
              $J_{ji}=-J_{ij}\trans$. At a differential Nash equilibrium,
              $\Ssym=\blockdiag(J_{11},\dots,J_{nn})\succ0$, so
              $\mu(\Metr,J)>0$ for every $\Metr\succ0$, and no SPD metric makes
              the equilibrium unstable.
              For zero-sum games and learning-rate ratios this is
              \cite[Thm.~3]{chasnov2021vector}.
        \item Let $n=2$ and $\Ssym_{12}\ne0$. At a differential Nash
              equilibrium, $\Ssym_{11}=J_{11}\succ0$ and $\Ssym_{22}=J_{22}\succ0$,
              and with the Schur complement
              $\Ssym/\Ssym_{11}:=\Ssym_{22}-\Ssym_{12}\trans\Ssym_{11}^{-1}\Ssym_{12}$,
              Haynsworth's formula gives
              $\operatorname{In}(\Ssym)=\operatorname{In}(\Ssym_{11})
                  +\operatorname{In}(\Ssym/\Ssym_{11})$, where
              $\operatorname{In}$ counts the positive, negative, and zero
              eigenvalues. Hence $\Ssym$ is indefinite if and only if
              $\lambda_{\min}(\Ssym/\Ssym_{11})<0$, equivalently
              \[
                  \bigl\|\Ssym_{11}^{-1/2}\Ssym_{12}\Ssym_{22}^{-1/2}\bigr\|_2>1:
              \]
              the off-diagonal block $\Ssym_{12}$ dominates the own curvature.
        \item Let $\Metr=\blockdiag(\Metr_1,\dots,\Metr_n)$ with
              $\Metr_i\in\R^{d_i\times d_i}$, $\Metr_i\succ0$, be a
              \emph{block-diagonal metric}: player $i$ preconditions only its own
              gradient, $\theta_i\mapsto\theta_i-\eta\Metr_i^{-1}\xi_i$. At a
              differential Nash equilibrium, every block-diagonal metric gives
              \[
                  \sum_{k=1}^m\Real\lambda_k(\Metr^{-1}J)
                  =\sum_{i=1}^n\tr\bigl(\Metr_i^{-1}J_{ii}\bigr)>0.
              \]
              Hence the unstable dimension is less than $m$, so a block-diagonal
              metric cannot make the equilibrium a source; if it makes the
              equilibrium hyperbolic and unstable, the result is a saddle. If
              moreover $m=2$ (two scalar players) and $\det J>0$, then
              $\mu(\Metr,J)>0$ for every block-diagonal metric: no block-diagonal
              metric destabilizes a strict sink, which extends
              \cite[Cor.~1(i)]{chasnov2020scalar} from learning-rate ratios to
              all block-diagonal metrics.
    \end{enumerate}
\end{theorem}
\begin{proof}
    (1) A block-diagonal matrix
    with positive-definite blocks is positive definite, so the left side of
    \eqref{eq:fov} is positive and every eigenvalue has positive real part. (2) Haynsworth's inertia additivity gives the
    first equivalence. For the second, with
    $C=\Ssym_{11}^{-1/2}\Ssym_{12}\Ssym_{22}^{-1/2}$,
    $\Ssym/\Ssym_{11}=\Ssym_{22}^{1/2}(I-C\trans C)\Ssym_{22}^{1/2}$ is
    congruent to $I-C\trans C$, which has a negative eigenvalue if and only if
    $\|C\|_2>1$. (3) With $Q=\Metr^{-1}$, $\tr(QJ)=\tr(Q\Ssym)+\tr(Q\Aanti)$
    and $\tr(Q\Aanti)=0$ for symmetric $Q$ and skew $\Aanti$;
    block-diagonality leaves $\sum_i\tr(Q_i\Ssym_{ii})$, and each term is the
    trace of a product of positive-definite matrices, hence positive. For
    $m=2$, $\det(QJ)=\det Q\det J>0$ and $\tr(QJ)>0$ place both eigenvalues in
    the open right half-plane.
\end{proof}

Noncooperation alone does not make $\Ssym$ indefinite. In a zero-sum pair,
$L_2=-L_1$ gives $J_{21}=-J_{12}\trans$, so $\Ssym_{12}=0$ and the interaction
is carried entirely by $\Aanti$; with identical interests, $L_1=L_2$ gives
$J_{21}=J_{12}\trans$ and $\Aanti=0$. Item~2 of \Cref{thm:nash} requires an off-diagonal block $\Ssym_{12}$ that dominates the own
curvature, which, for example, a common
collision cost can produce. For zero-sum differential Nash
equilibria, setting $M=I$ recovers the gradient-flow stability result of
Mazumdar et al.~\cite{mazumdar2020gradient}.

\subsection{From local destabilization to equilibrium selection}
\label{sec:statedep}
The construction has two steps. First, a constant SPD matrix is chosen for
each of the two equilibria: $\Ms$ for the selected equilibrium, which must
keep it attracting, $\mu(\Ms,J(\ths))>0$, and $\Mr$ for the rival, which must
make it unstable, $\mu(\Mr,J(\thr))<0$. For a rival that is initially a sink, \Cref{thm:switch} gives one choice
of $\Mr$ whenever $\Ssym(\thr)$ is indefinite, and \Cref{sec:placement} describes how both
matrices are constructed in practice. Second, the two matrices are joined into
a single metric field that takes the value $\Ms$ at $\ths$ and $\Mr$ at
$\thr$, so $\Ms$ and $\Mr$ are called the \emph{equilibrium metrics}. The field
is
\begin{equation}
    M(\theta)=w(\theta)\Ms+(1-w(\theta))\Mr,\qquad 0\le w\le1.
    \label{eq:statedep}
\end{equation}
For distinct equilibria, one globally smooth weight that reproduces the two
equilibrium metrics exactly is
\begin{equation}
    w(\theta)=\frac{\|\theta-\thr\|^2}
    {\|\theta-\ths\|^2+\|\theta-\thr\|^2}.
    \label{eq:exactweight}
\end{equation}
It gives $M(\ths)=\Ms$ and $M(\thr)=\Mr$, and, like every SPD metric field,
it preserves every zero. All experiments use this interpolant.

\begin{theorem}[Local selection by a state-dependent metric]
    \label{cor:selection}
    If distinct $\ths,\thr$ are strict sinks under $M=I$ and $S(\thr)$ is
    indefinite, a smooth SPD metric field preserves all equilibria while leaving
    $\ths$ a strict sink and making $\thr$ linearly unstable.
\end{theorem}
\begin{proof}
    Set $\Ms=I$, take $\Mr$ to be the metric $\Metr_{-}$ of \Cref{thm:switch}, and use
    \eqref{eq:statedep}--\eqref{eq:exactweight}. The linearizations at the two
    equilibria have the required spectra.
\end{proof}

Positive definiteness of the interpolant guarantees invertibility, not
stability throughout the state space; in gain scheduling, stability-preserving
interpolation requires additional conditions~\cite{stilwell2000interpolation}.

\paragraph{Selecting between symmetry-related equilibria}

A game that treats its players identically can have equilibria that differ
only in which player takes which role. In the iterated prisoner's dilemma of
\Cref{sec:dilemmas}, neither attracting equilibrium is symmetric: in each, one
player cooperates more and the other earns more, and exchanging the players
maps one equilibrium onto the other. Such a pair is called a
\emph{symmetry-related} pair. Formally, a symmetry of the game is an orthogonal
matrix $R$ (for example, the player swap $\Pi$) such that
$\xi(R\theta)=R\,\xi(\theta)$.
Differentiating gives $J(R\theta)=RJ(\theta)R\trans$, so if $\thr=R\ths$, then
$J_r=RJ_sR\trans$. Because $R\trans=R^{-1}$, this is a similarity
transformation: $J_r$ and $J_s$ have the same eigenvalues, and gradient play
attracts both equilibria or neither.

The same conclusion holds for any update rule that does not distinguish the
two roles. Consider an update $\theta_{k+1}=\theta_k-\eta\,u(\theta_k)$ whose
direction $u$ vanishes at the equilibria. With $\delta_k=\theta_k-\theta^\star$,
its linearization is $\delta_{k+1}=(I-\eta K)\delta_k$, where
$K=Du(\theta^\star)$ is the Jacobian of the update direction. Gradient
play ($u=\xi$) gives $K=J$, a constant metric ($u=\Metr^{-1}\xi$) gives
$K=\Metr^{-1}J$, SGA ($u=\xi+\lambda\Aanti\trans\xi$) gives
$K=(I+\lambda\Aanti\trans)J$, and consensus optimization
($u=\xi+\gamma J\trans\xi$) gives $K=(I+\gamma J\trans)J$. Suppose $K=F(J)$ is
computed from $J$ alone. The rule is \emph{equivariant} if relabeling the
coordinates of its input relabels its output in the same way,
$F(RJR\trans)=RF(J)R\trans$; informally, the rule has no built-in preference
between the two roles.

\begin{theorem}[Equivariant updates]
    \label{thm:equivariance}
    If $J_r=RJ_sR\trans$ for an orthogonal $R$ and $F$ is equivariant, then
    $K_r=RK_sR\trans$. Hence $K_r$ and $K_s$ have identical spectra, and $\thr$
    is attracting under the update exactly when $\ths$ is.
\end{theorem}
\begin{proof}
    $K_r=F(RJ_sR\trans)=RF(J_s)R\trans=RK_sR\trans$.
\end{proof}

Gradient play, and SGA and consensus optimization with fixed scalar gains, are
equivariant, as is any metric computed equivariantly from $J$; methods with
internal state, such as momentum, are equivariant when that state transforms
in the same way. Selecting one equilibrium of a symmetry-related pair therefore
requires a convention that breaks the symmetry, such as which player
concedes. The state-dependent field \eqref{eq:statedep} provides one: it
assigns a stabilizing metric to the selected equilibrium and a destabilizing
metric to its symmetric counterpart.

\section{Metric synthesis and execution}
\label{sec:estimation}

This section constructs the equilibrium metrics from the Jacobians at the equilibria,
joins them by the interpolant \eqref{eq:exactweight}, and executes the
resulting flow with explicit Euler steps. The equilibria $\ths,\thr$ are
assumed known, and $\Jhat$ denotes the Jacobian supplied to the design,
either estimated or computed exactly.

\subsection{Constructing candidate equilibrium metrics}
\label{sec:placement}

The rank-one construction of \Cref{thm:switch} gives a direct candidate for
the rival: choose an eigenvector of $\Msym(\widehat J_r)$ with a negative
eigenvalue and increase $t$ in $Q_t=I+t vv^\top$. If gradient play already
attracts the selected equilibrium, $M_s=I$ is a candidate stabilizing metric.
For either candidate, the spectra at the equilibria are then computed to confirm the
signs that selection requires: $\mu(\Ms,\widehat J_s)>0$ at the selected
equilibrium and $\mu(\Mr,\widehat J_r)<0$ at the rival. No margin is
prescribed for these candidates; margins enter only through the fitted
metrics below.

The rank-one candidate fixes only the sign of the rival's slowest mode. It
offers no control over how quickly the selected equilibrium attracts or how
quickly the rival repels, and the direction $v$ generally mixes the players'
coordinates, so it cannot respect a requirement that each player rescale only
its own update, that is, a block-diagonal metric. To meet prescribed margins
or such player-block constraints, the experiments instead fit
$\eig(\Metr^{-1}\widehat J)$ numerically, treating the eigenvalues at the equilibria as
design targets. A metric must remain symmetric positive definite, which an
unconstrained optimizer would not preserve. The metric is therefore
parameterized on the positive-definite cone by a symmetric log-parameter $W$,
optionally masked to block-diagonal form:
$\Metr(W)=\exp(\Msym(W)\odot\mathrm{mask})$, where $\exp$ is the matrix
exponential and the block mask is symmetric. Every $W$ then yields a valid
metric, and the mask enforces the player structure exactly rather than by
penalty. The two equilibria place opposite demands, since the selected
equilibrium must remain attracting while the rival must become repelling.
Because the interpolant \eqref{eq:exactweight} applies each equilibrium metric
only at its own equilibrium, these demands decouple, and the two equilibrium
metrics are fitted separately.
Write $K_j(W)=\Metr(W)^{-1}\widehat J_j$, $j\in\{s,r\}$, for the
Jacobian of the update direction $\Metr(W)^{-1}\xi$ at each equilibrium and $\lambda_k(K)$ for the $k$th
eigenvalue of a matrix $K$, as in \eqref{eq:margin}. The objectives are
\begin{equation}
    \begin{aligned}
        \ell_{s}(W) & =\operatorname{mean}_k\operatorname{softplus}\bigl(\alpha_s-\Real\lambda_k(K_s(W))\bigr) \\
                    & \quad+\rho_c\|\Msym(W)\|_F^2,                                                            \\
        \ell_{r}(W) & =\operatorname{softplus}\bigl(\alpha_r+\mu(\Metr(W),\widehat J_r)\bigr)                  \\
                    & \quad+\rho_c\|\Msym(W)\|_F^2,
    \end{aligned}
    \label{eq:surrogate}
\end{equation}
where $\operatorname{softplus}(x)=\log(1+e^x)$ and $\alpha_s,\alpha_r>0$
are target margins. Minimizing $\ell_s$ pushes the real parts of all
eigenvalues at the selected equilibrium above $\alpha_s$, and minimizing
$\ell_r$ pushes the stability margin \eqref{eq:margin} at the rival, its
smallest real part, below $-\alpha_r$.
The signs follow from the linearization $-M^{-1}J$: eigenvalues of $M^{-1}J$
with positive real part give decay, whereas one with negative real part gives
a growing mode. The two objectives differ in form because the two requirements
differ. Attraction requires every eigenvalue to have positive real part, so
$\ell_s$ penalizes each eigenvalue that falls below the margin; repulsion
requires only a single growing mode, so $\ell_r$ acts on the smallest real
part alone and leaves the remaining modes free. The softplus is a smooth
surrogate for $\max(0,x)$, which supplies gradients where the hinge would be
flat or kinked. The term weighted by $\rho_c$ penalizes large log-metrics and
thereby keeps the fitted metrics well conditioned; a poorly conditioned metric
can enlarge the eigenvalue magnitudes $|\lambda_k|$ and so shrink the
admissible step size \eqref{eq:stepsize}.
Because these penalties are soft, they do not guarantee the target margins,
and the fitted spectra at the equilibria must be checked.
Both objectives are continuous in $W$; they are differentiable except where
eigenvalues collide or, for $\ell_r$, where two eigenvalues tie for the
smallest real part. Normalizing the trace after fitting changes the
rates but not the signs. The rival fit must be initialized with care. When
the rival's eigenvalues under $\Metr=I$ are real, decreasing $\mu$ from
$W=0$ drives the smallest eigenvalue toward zero, which it cannot cross along
an SPD path (\Cref{thm:index}), so the fit can stall at a positive margin. The
rival fit is therefore started at the rank-one metric of \Cref{thm:switch},
whose log-parameter is $W_0=-\log(1+t)\,vv\trans$ and whose trace
$\tr(\Metr^{-1}\widehat J_r)$ is already negative; under a block mask, $W_0$
is masked, and the fitted margin is checked as before. The fitting objective uses the full $J$ but
provides no certificate. A certificate for a fixed $\Metr$ follows from
bisection on $\alpha$ over the standard Lyapunov LMI
\begin{equation}
    \text{find }P:\ I\preceq P\preceq\kappa I,\ \
    P(\Metr^{-1}J)+(\Metr^{-1}J)\trans P\succeq2\alpha P
    \label{eq:certify}
\end{equation}
which returns a lower bound on $\mu$; the conditioning cap $\kappa$ can
make this bound conservative.

\subsection{Discrete-time execution}

At an exact zero, the derivative of $\Metr^{-1}(\theta)$ multiplies
$\xi(\theta^\star)=0$, so the linearization of the Euler update is
$I-\eta\Metr^{-1}J$, and the standard local convergence criterion is
$\rho(I-\eta\Metr^{-1}J)<1$ (see Chasnov et
al.~\cite[Prop.~2 and Remarks~1--2]{chasnov2020convergence}). With
$\lambda_i\in\eig(\Metr^{-1}J)$, this holds exactly for step sizes
\begin{equation}
    0<\eta<\eta_{\max}:=\min_i
    \frac{2\Real\lambda_i}{|\lambda_i|^2},
    \label{eq:stepsize}
\end{equation}
provided every $\Real\lambda_i>0$; otherwise set $\eta_{\max}=0$.
The criterion does not require diagonalizability, and it gives local
exponential convergence in a suitable norm, not necessarily the Euclidean
one.

If the rival has an eigenvalue with $\Real\lambda<0$, then
$|1-\eta\lambda|>1$ for every $\eta>0$, so strict instability in continuous
time carries over to the Euler linearization.

The bound \eqref{eq:stepsize} can be restrictive: increasing $|\Imag\lambda|$
at fixed $\Real\lambda$ shrinks $\eta_{\max}$ quadratically.

\paragraph{Decentralized execution and centralized design}
A block-diagonal metric splits the update into per-player steps
$\theta_{i,k+1}=\theta_{i,k}-\eta\,\Metr_i(\theta_k)^{-1}\xi_i(\theta_k)$,
where $\Metr_i(\theta)=w(\theta)\Metr_{s,i}+(1-w(\theta))\Metr_{r,i}$ is player
$i$'s block of \eqref{eq:statedep}. Spectral synthesis from the
full Jacobian, however, requires centralized design, and the interpolation
weight may itself require shared state, so per-player linear solves alone do
not make sensing or computation fully decentralized. At differential Nash
equilibria, block-diagonal metrics cannot create hyperbolic sources, although
in higher dimensions they can create saddles (\Cref{thm:nash}(3)).

\section{Experiments}
\label{sec:experiments}

Two social games test the structural destabilization conditions and empirical selection under
a common protocol. Throughout, local stability is reported separately from
basin coverage, and execution uses exact, noise-free field evaluations.

\subsection{Protocol}
Equilibrium metrics are fitted to finite-difference estimates of the
Jacobians at the equilibria by minimizing \eqref{eq:surrogate} with
$(\alpha_s,\alpha_r)=(0.30,0.15)$, normalized to $\tr M=m$, and joined by the
exact interpolant \eqref{eq:exactweight}. \Cref{tab:outcomes} reports the
outcomes of the resulting metric and, for reference, of unmodified gradient
play from the same initializations and budgets; a run counts toward an
equilibrium when it ends within $10^{-3}$ of it, and \emph{other} denotes
escape or non-convergence. The basin outcomes are empirical.

\begin{table}[t]
  \centering
  \caption{Outcomes as selected/rival/other. Gradient play is the
    unmodified dynamics, shown for reference.}
  \label{tab:outcomes}
  \setlength{\tabcolsep}{3pt}
  \footnotesize
  \begin{tabular}{@{}lcc@{}}
    \toprule
    Update                   & Stag hunt ($81$)    & Iterated PD ($120$) \\
    \midrule
    Gradient play            & $36/45/0$           & $60/60/0$           \\
    Metric, full             & $\mathbf{81/0/0}$   & $\mathbf{120/0/0}$  \\
    Metric, block-diagonal   & cannot destabilize & $\mathbf{120/0/0}$  \\
    \bottomrule
  \end{tabular}
\end{table}


\subsection{Social dilemmas: stag hunt and iterated prisoner's dilemma}
\label{sec:dilemmas}

\begin{figure*}[t]
      \centering
      \includegraphics[width=\textwidth]{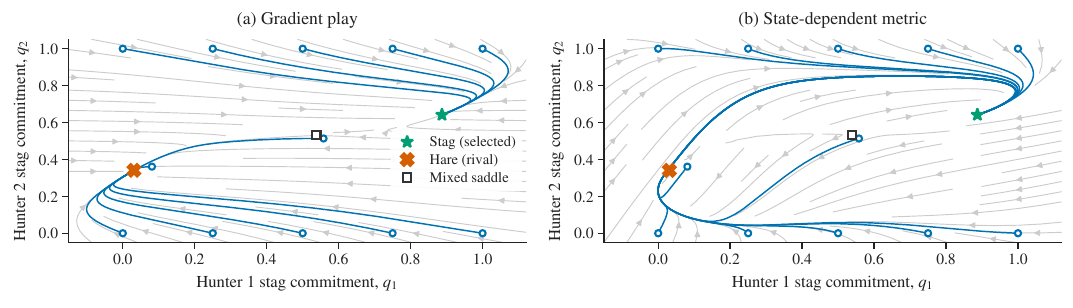}
      \caption{Stag hunt: trajectories from the same twelve initial commitments
            (open circles) under (a) gradient play and (b) the state-dependent
            metric, with $\eta=0.01$ and at most $20\,000$ updates. Gradient play
            reaches either equilibrium depending on the start; the metric keeps
            both equilibria in place, makes hare unstable, and every displayed
            trajectory reaches stag. Outcomes for all $81$ starts are reported in
            Table~\ref{tab:outcomes}.}
      \label{fig:staghunt}
\end{figure*}

\begin{figure*}[t]
      \centering
      \includegraphics[width=\textwidth]{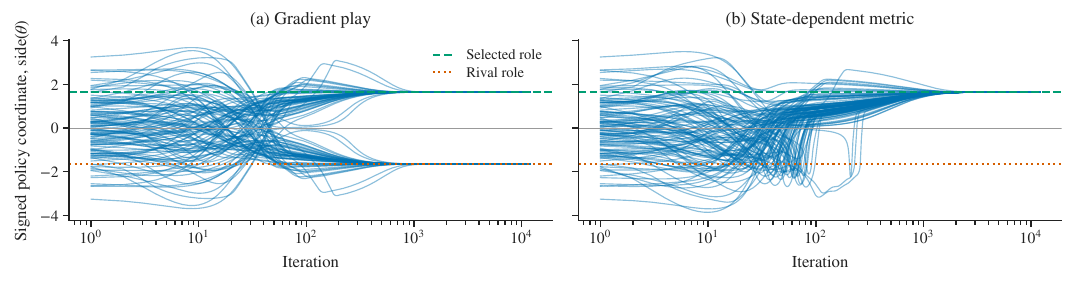}
      \caption{Iterated PD: signed policy coordinate along the axis from the
            rival to the selected equilibrium, for all $120$ starts, under
            (a) gradient play and (b) the state-dependent metric. Gradient play
            splits $60/60$ between the two roles; under the metric, runs that
            approach the rival are repelled, and all $120$ reach the selected
            role.}
      \label{fig:ipd}
\end{figure*}

Smooth extensions of two classical dilemmas test the structural results. The
continuous-commitment stag hunt has two equilibria that are \emph{not}
related by a symmetry: stag is payoff-dominant and hare is risk-dominant. The
entropy-regularized iterated prisoner's dilemma has a pair of
symmetry-related equilibria, exchanged by swapping the players.

\paragraph{Stag hunt}
Each hunter chooses a stag commitment $q_i\in\R$, where $q_i=1$ is pure stag
and $q_i=0$ is pure hare. The stag can be caught only jointly and pays
$R\,q_iq_j$; the hare can be caught alone and pays $P(1-q_i)$ regardless of
the partner. Hunter $i$ weights the catch by a stake $\kappa_i>0$, pays a
convex commitment cost, and minimizes
\begin{equation}
      L_i=\gamma\bigl(q_i-\tfrac12\bigr)^4-\kappa_i\bigl(Rq_iq_j+P(1-q_i)\bigr),
      \label{eq:staghunt}
\end{equation}
with $R=3$, $P=1.6$, $\gamma=1.4$, and $(\kappa_1,\kappa_2)=(1,0.015)$. At
pure commitments $q\in\{0,1\}^2$, the reward reduces to the classic stag-hunt
matrix~\cite{skyrms2004stag}: $R$ to each hunter if both pursue the stag, $P$
to a hare hunter regardless of the partner, and nothing to a lone stag
hunter. The quartic cost makes the best responses nonlinear; a quadratic cost
would give an affine field and hence a unique equilibrium. With these
parameters, stag at $q^\star=(0.89,0.64)$ and hare at $(0.03,0.34)$ are
attracting differential Nash equilibria under gradient play, separated by a
mixed saddle at $(0.54,0.53)$, and stag gives both hunters lower loss. Both
equilibria lie in $[0,1]^2$, so each coordinate reads as a level of
commitment: the high-stakes hunter commits almost fully to its prey, whereas
the low-stakes hunter only leans toward it. At both equilibria $\Ssym$ is
indefinite and $\Aanti\ne0$, as \Cref{thm:switch} requires.

The fitted metric keeps stag a sink, with eigenvalues $0.89\pm0.14i$. It
makes hare a source, with eigenvalues $-31.1$ and $-5.1$ and
$\operatorname{tr}(M(\thr)^{-1}\Ssym)=-36.3$. The mixed equilibrium remains a
saddle.

Execution uses $\eta=0.01$ and at most $20\,000$ updates from a $9\times9$
grid of initial commitments in $[0,1]^2$. The step size satisfies \eqref{eq:stepsize} at stag, where
$\eta_{\max}=2.20$. The metric reaches stag from all
$81$ starts, whereas unmodified gradient play reaches stag from $36$ and hare
from $45$. \Cref{tab:outcomes} lists the outcomes, and \Cref{fig:staghunt}
shows representative trajectories.

Because stag and hare are not related by a symmetry, \Cref{thm:equivariance}
does not rule out a fixed metric. In two dimensions, however, the argument of
\Cref{thm:nash}(3) shows that stability is decided by the sign of
$\operatorname{tr}(\Metr^{-1}\Ssym)$. Since
$\Ssym_{\mathrm{stag}}\preceq\Ssym_{\mathrm{hare}}$, any fixed $\Metr$ that
destabilizes hare also destabilizes stag, and by \Cref{thm:nash}(3) no
block-diagonal metric can destabilize hare.

\paragraph{Iterated prisoner's dilemma}
Memory-one strategies induce a four-state Markov chain in the repeated
prisoner's dilemma~\cite{press2012ipd}. Here player $i$ holds logits
$\theta_i\in\R^5$ for its cooperation probabilities $\sigma(\theta_{i,k})$, as in
differentiable opponent-learning models~\cite{foerster2018lola,letcher2019sos},
and $\theta=(\theta_1,\theta_2)\in\R^{10}$ is the joint policy.
The five entries cover the first move and the moves after CC, CD, DC, and
DD. A player earns the
reward $R$ for mutual cooperation, the sucker's payoff $S$ for cooperating
against a defector, the temptation $T$ for defecting against a cooperator, and
the punishment $P$ for mutual defection, so the payoff vector of player $i$
over the states CC, CD, DC, DD is $r_i=(R,S,T,P)$. With
$(R,S,T,P)=(3,0,5,1)$, which satisfies the dilemma ordering $T>R>P>S$,
discount $\gamma=0.95$, and entropy weight $\tau=0.2$, the loss of player $i$
as a function of the joint policy is
\begin{equation}
      L_i(\theta)=-(1-\gamma)p_0^\top(I-\gamma P_\theta)^{-1}r_i
      -\tau\sum_{k=1}^5H(\sigma(\theta_{i,k})),
      \label{eq:ipd}
\end{equation}
where $p_0$ is the joint initial distribution, $P_\theta$ the memory-one
transition matrix induced by $\theta$, $\sigma$ the logistic function,
and $H$ the binary entropy, summed over the policy entries without weighting
by state occupancy. The differentiable formulation
follows~\cite{foerster2018lola,letcher2019differentiable,letcher2019sos}, and
the entropy term defines the regularized game studied here. Gradient play has
two attracting equilibria in which one player cooperates more and the other
earns more, with values $(V_1,V_2)=(2.04,1.22)$ and, with the players exchanged,
$(1.22,2.04)$. They are related by the block swap $\Pi$ of the players' logits. A symmetric saddle with one unstable real mode separates
them. Since $\Pi$ is orthogonal, \Cref{thm:equivariance} implies that every
fixed equivariant rule assigns both equilibria the same spectrum. Selection is
thus a role convention that decides which player leads. Both equilibria are
differential Nash since the smallest own-Hessian eigenvalues are positive as required by \eqref{eq:differentialnash}.
$\Ssym$ is indefinite, with least eigenvalue $-0.014$ and least Schur-complement
eigenvalue $-0.11$. The antisymmetric coupling $\|\Aanti_{12}\|_2=0.076$
exceeds the symmetric coupling $\|\Ssym_{12}\|_2=0.031$. Under $\Metr=I$ the
spectrum contains the rotational pair $0.035\pm0.064i$. No isolated
mutual-cooperation equilibrium was found at these parameters.

The full metric keeps the selected equilibrium a sink and makes the rival a
saddle with $\mu_r=-0.466$ and $\operatorname{tr}(\Mr^{-1}\Ssym)=-0.15$.
Exactly two modes become unstable and the other eight modes stay stable along the SPD path, which is consistent
with \Cref{thm:index}. Execution uses $\eta=2$ and at most
$12\,000$ updates from $120$ initial policies: $60$ drawn from
$\mathcal N(0,1.5^2I)$ and their images under $\Pi$. Gradient play splits
exactly $60/60$, as \Cref{thm:equivariance} predicts for an update that does
not distinguish the two roles (\Cref{fig:ipd}(a)). The metric reaches the
selected equilibrium from all $120$ starts (\Cref{fig:ipd}(b)). A \emph{block-diagonal} metric also destabilizes the rival. It
gives $\mu_r=-0.342$ with two unstable real modes and selects from all $120$
starts. Since
$\operatorname{tr}(\Mr^{-1}\Ssym)=+0.95>0$, the rival becomes a saddle rather
than a source, as \Cref{thm:nash}(3) allows: in ten dimensions, a block-diagonal
metric has room to create a saddle.

\section{Conclusion}
\label{sec:conclusion}
This work uses classical $H$-stability results to select equilibria in differentiable
games: at an equilibrium that some SPD metric stabilizes, a strictly
destabilizing SPD metric exists exactly when the symmetric part of the game
Jacobian is indefinite, which turns equilibrium selection among known
equilibria into a condition on the local game structure. A smooth interpolation between the equilibrium metrics at the selected and rival equilibria then turns the local
designs into a metric field that preserves all zeros. In the stag hunt and
the iterated prisoner's dilemma, where gradient play splits between the two
equilibria, the resulting metric reached the selected equilibrium from every
tested start while leaving all equilibria in place. Future work will apply this construction to trajectory games to
encode traffic conventions, such as the give-way rules of car and ship
encounters.

\bibliographystyle{IEEEtran}
\bibliography{references}

\end{document}